\documentclass[12pt]{article}

\usepackage{amsmath, amssymb, graphicx, url, algorithm2e}

\usepackage{times}
\usepackage{helvet}
\usepackage{courier}

\usepackage{amstext, dsfont}
\usepackage{amstext, amsthm, dsfont}

   \theoremstyle{plain}
     \newtheorem{theorem}{Theorem}
\newtheorem{claim}{Claim}
     \newtheorem{lemma}[theorem]{Lemma}

   \theoremstyle{definition}

   \newcommand{\reals}{\mathbb{R}}

   \newcommand{\argmin}{\mathbf{ArgMin}}

   \renewcommand{\O}{{\mathcal O}}

  \newcommand{\Scal}{{\mathcal S}}

\title{Computational Feasibility of Clustering under Clusterability Assumptions\\
{\small \textrm{``Vanity of vanities, saith Koheleth; vanity of vanities, all is vanity."}(Ecclesiastes, Ch1;2)}}

\usepackage{times}

\author{Shai Ben-David\\
 School of Computer Science,
University of Waterloo,\\ Waterloo,Ontario, N2L 3G1,
Canada\\
{\tt shai@uwaterloo.ca}}

\begin{document}

\maketitle

\begin{abstract}

It is well known that most of the common clustering objectives are NP-hard to optimize.  In practice, however, clustering is being routinely carried out. 

One approach for providing theoretical understanding of this seeming discrepancy is to come up with 
notions of clusterability that distinguish realistically interesting input data from worst-case data sets. The hope is that  there will be 
clustering algorithms that are provably efficient on such ``clusterable" instances. In other words, hope that ``Clustering is difficult only when it does not matter"\footnote{This phrase is in fact a title of a recent paper --  \cite{DLS}.} (the \emph{CDNM thesis} for short).
%In other words, one cares only about clustering "clusterable" data, and there maybe clustering algorithms that on such data find optimal clustering efficiently.  
We believe that to some extent this may indeed be the case. This paper provides a survey of recent papers along this line of research and a critical evaluation their results. Our bottom line conclusion is that  that CDNM thesis is still far from being formally substantiated.

We start by discussing which requirements should be met in order to provide formal support the validity of the CDNM thesis. In particular, we 
 list some implied requirements for notions of \emph{clusterability}. We then examine existing results in view of these requirements and outline some research challenges and open questions.\\

\end{abstract}

\section{Introduction}
The goal of this note is two-fold. First, I would like to provide a personally biased overview of the research concerning the computational complexity of clustering under data niceness assumptions.  Having worked in this area for quite some time now, I feel that while the TCS community appreciates 
work that may have practical relevance (and clustering is clearly a task that arises in many applications),  sometimes in this area there is a significant gap between research motivation and the actual technical results it yields. A secondary aim of this paper is to call the attention of the theoretical research community to some such gaps and encourage further work along directions that might  have otherwise seemed resolved.

\subsection{Alternatives to worst-case for measuring computational complexity}
Computational complexity theory aims to provide tools for the
quantification and analysis of the computational resources needed
for algorithms to perform computational tasks. Worst-case complexity
is by far the best known, most researched and best
understood approach to computational complexity theory.  In particular, NP-hardness is a worst-case-instance notion. By saying that a task is NP-hard (and assuming $P \neq NP$), we imply that for every algorithm, there exist infinitely many instances on which it will have to work hard. However, for many problems 
this measure is unrealistically pessimistic compared to the experience of solving them for practical instances. A problem may be NP--hard and still have algorithms that solve it efficiently for any instance that is likely to occur in practice or any instance for which one cares to find an optimal solution for.

Several approaches have been proposed to bringing computational complexity theory closer to the actual hardness faced when solving optimization problems on real data. \emph{Average Case Complexity} (\cite{Levin86}, \cite{Ben-DavidCGL89}), analyzes run time w.r.t. some given probability distribution over the input instances.   \emph{Smoothed Analysis}  (\cite{SpielmanT01}) examines the running time of a given algorithm by taking the worst case over all inputs of the average runtime of the algorithm over some vicinity of the input. A different approach is to have a notion of ``well-behaved-instances", so that on one hand it is reasonable to expect that instances one comes across in applications are so well behaved, and on the other hand there exist algorithms that can solve any well behaved input in polynomial time. Various earlier approaches have addressed computational hardness by defining subset of relatively-easy instances (most notably, the area of \emph{parameterized complexity} (\cite{DowFell98})).
\cite{Ben-David06},  and  \cite{BiluL10} propose general notions of tamed instances that apply across different problems. Both of these papers apply some type of \emph{robustness to perturbations} as the key property of such well behaved instances. Algorithms that efficiently solve  NP-hard problems on such perturbation robust instances have been shown to exist for agnostic learning of half-spaces (\cite{BDUS}) and for graph partitioning problems ( \cite{BiluL12}).   \cite{BalcanBG09}
formalized a  \emph{ uniqueness of the optimal solution} criterion as a notion of well behaved clustering instances, which can also be applied to other types of problems. In this note we will focus on the application of such approaches to clustering. We will discuss those, as well as other notions of niceness-of-instances that are specific to clustering problems, as a basis for alternatives-to-worst-case-complexity analysis of clustering tasks.

\subsection{A focus on clustering tasks}
Clustering is a very useful paradigm that is being applied in a wide range of data exploration tasks. The term ``clustering" should be thought of as an umbrella notion for a big and varied collection of tasks and algorithmic paradigms. Here, we focus on clustering tasks that are defined as discrete optimization problems. Most of those optimization problems are NP-hard.
We wish to examine whether this hardness remains an issue when we restrict our attention to ``clusterable data" - data for which a meaningful clustering exists (one can argue that when there is no cluster structure in a given data set, there is no point in applying a clustering algorithm to it). In other words, we wish to evaluate to what extent current theoretical work supports the ``Clustering is difficult only when it does not matter" (CDNM) thesis.

For the sake of concreteness, we will focus on two popular clustering objectives, $k$-means and $k$-median.

\subsection{Outline of the paper}
We start this note by listing, in Section \ref{requirements},  what we think are requirements from notions of clusterability aiming to substantiate the CDNM thesis. In Section \ref{notions_clust}, we list various notions of clusterability that have been proposed in the context of this line of research. These include: 
\emph{Additive perturbation robustness} (APR), \cite{AckermanB09}; 
\emph{Multiplicative perturbation robustness} (MPR), \cite{BiluL10};
$(\alpha, \epsilon)$ \emph{Perturbation Resilienc}e, \cite{BalcanL12};
\emph{$\epsilon$ -Separatedness}, \cite{OstrovskyRSS12};
\emph{Uniqueness of optimum},  \cite{BalcanBG09} (they call it $(c, \epsilon)$\emph{-approximation-stablility});
\emph{$\alpha$-center stability}, \cite{AwasthiBS12};
and \emph{$(1+\alpha)$ Weak Deletion Stability}, \cite{AwasthiBS10}.

The main body of this paper is an examination, in Section \ref{meet_req}, of how well do the current notions and results meet the requirements (of Section \ref{requirements}).  To get a sense of how strict a clusterability condition is, we consider an optimal clustering of data sets that satisfy that condition and examine the implied bounds on the ratio between the average distance of a data point to its own cluster center and the distance between centers of different clusters (or the distance of a point from centers of clusters it does not belong to). By analyzing the results pertaining to the proposed notions of clusterability listed above, we show, for example, that,  
\begin{itemize}
\item The values of $\epsilon$ for which $\epsilon$ -Separatedness is shown (in \cite{OstrovskyRSS12}) to allow $poly(k)$ clustering algorithms imply that, in the optimal clustering,  the average distance of a point from its cluster center
should be smaller than the minimal distance between distinct cluster centers by a factor of at least 200.
\item The values of parameters for which $(c, \epsilon)$ approximation stability  is shown (in \cite{BalcanBG09}) to allow $poly(k)$ clustering algorithms imply that, in the optimal clustering, for all but an $\epsilon$-fraction of the input points, the distance of a point to its own cluster center is smaller than its distance to the next closest center by at least 20 times the average point-to-its-cluster-center-distance.

\item The values of $\alpha$ for which $(1+\alpha)$ weak deletion stability is shown (in \cite{AwasthiBS10}) to allow $poly(k)$ clustering algorithms imply that, in the optimal clustering, the vast majority of the clusters are so distant from the rest of the data points that any point outside such a cluster is further from the center of that cluster by at least $\log(k)$ times the "average radius" of its own cluster.

\end{itemize}

Our conclusion is that the currently available theory is still far from substantiating the CDNM thesis. In particular, while additive perturbation robustness, with any non-zero robustness parameter, gives rise to algorithms that find the optimal clusterings in time polynomial in the input size and its dimension, as far as currently published results go, non of the requirements listed above allows finding optimal clustering solutions in time polynomial in the number of target clusters, $k$, unless the corresponding parameters are set to values that hold only for extremely well clusterable data sets\footnote{The above consequences 
of the required clusterability conditions 
are obtained by examining the parameter values and constants that are implicit in the asymptotic formulation of the 
efficiency results in the above cited papers. One should note that these negative statements reflect only the current state of knowledge, and are not proven lower bounds.
For some of the above notions of clusterability, we also discuss lower bounds on the parameter values required to overcome the NP-hardness of the clustering tasks.}.

In Section \ref{conclusions} we discuss these discouraging results further, highlight some implied open problems and propose directions in which this line of research should, in our opinion, proceed.

\section{Requirements from notions of clusterability} \label{requirements}
We begin this discussion by  stating requirements that (we believe) a notion of clusterability should satisfy to be applied for supporting the ``Clustering is Difficult only when it does Not Matter" (CDNM, in short) thesis. At this point those requirements will be stated as qualitative, high level, statements. We discuss more concrete quantitative formulations in Section \ref{meet_req} .

\begin{enumerate} 

\item \emph{ It should be reasonable to assume that most (or at least a significant proportion of) the inputs one may care to cluster in practice satisfy the clusterability notion.}\\

 Some disclaimer is in place here; Of course, we do not have any way to guarantee that unseen practical instances will satisfy any non-trivial requirement. 
However,  this type of consideration can serve as a way to filter out clusterability conditions that are too restrictive. Furthermore, when a good data generative model is available, one can formalize requirements pertaining to a high probably of having the generated instances satisfy the given clusterability notion. 
 
\item  \emph{ In order to support the CDNM thesis, a notion of clusterability should be such  that there exist efficient algorithms that are guaranteed to find a good clustering (minimizing the objective function, or getting very close to it) for any 
input that satisfies that clusterability requirement}.
\end{enumerate} 
The next two requirements may be more debatable. Their significance is motivated by considering practical aspects of clustering applications.
Assume we do have some clusterability condition and a guarantee that the algorithm we are about to run is efficient on instances satisfying it. Still, when we get some real input, there is no guarantee that it satisfies that clusterability condition. If it does not, and we run our algorithm, it may either run for too long or terminate with some sub-optimal solution. However, for most of the NP-hard clustering problems, there is no efficient way of measuring how far from optimal  a given clustering solution is. We are therefore in the risk of not being able to protect against bad solutions. This consideration implies
a third desirable requirement  -- the ability to distinguish between clusterable and non-clusterable input data sets. Namely,

3. \emph{ There exists an efficient algorithm for testing clusterability. Namely, given an instance $(X,d)$, the algorithm determines whether it satisfies the clusterability requirement or not.}

Another advantage of having a notion of clusterability satisfy this requirement is that it will allow a direct evaluation of the extent to which the notion satisfies Requirement 1 above. Namely, having an efficient clusterability -checking algorithm, one could apply it to collections of representative practical clustering inputs from various domain and evaluate to what extent the clusterability requirement actually holds for such clustering tasks.\\

A forth, somewhat orthogonal, desiderata relates to existing common clustering algorithms. Namely, \\

4. \emph{ Some commonly used clustering algorithm can be guaranteed to perform well (i.e., run in polytime and find close-to-optimal solutions) on all instances satisfying the clusterability assumption.}\\

Requirement 4 is important if our goal is to \emph{understand} what is happening nowadays in clustering work by providing a theoretical explanation for the success of common clustering algorithms on real data. However, even when failing it, requirement 2 may lead to the development of new clustering algorithms, which may have independent merits.\\

\noindent {\bf  The main Open Question:} \emph{Find a notion of clusterability that satisfies the requirements above (or even just the first two).}\\

\section{Definitions and basic notions}
We consider clustering tasks that can be described as follows:
\begin{itemize}
\item The input is a finite subset $X$ of a metric space $(Y,d)$ \footnote{In some cases, $d$ is not required to satisfy the triangle inequality, in which case we call it a \emph{dissimilarity function} rather than a metric.},
and some number $k$. When $X=Y$ or $Y$ is  some Euclidean space (with the Euclidean distance), we omit mentioning it explicitly. 
\item The solution space $\Scal$ is a collection of partitionings  of the input set $X$ into $k$ subsets (a.k.a. \emph{$k$-clusterings}).
\item The problem is determined by an \emph{objective function}, $\O$,  that maps pairs  of (Instance, Clustering) to the real numbers. The goal of the algorithm is to find a clustering in the solution space that minimizes this objective for the given input instance. We let $C_{\O} (X,d)$ be $\argmin_{C \in \Scal} \O(C, (X,d))$ (namely the set of all clusterings in the solution space that minimize the objective cost for the input). Finally, given some objective function, let $OPT(X,d)$ denote the cost of an optimal clustering or $(X,d)$ (this value, depends, of course, on the objective function in question. However, to simplify the notation, we suppress this dependence on the objective). We also suppress the distance function $d$ when it is clear from the context and when it is the Euclidean distance.
\end{itemize}
We zoom in even further and consider only  ``center based" clustering objectives. For such problems, a clustering is defined by a set of $k$ points (centers), the partition associated with such a set of centers is the Voronoi partition it induces over the input set, and the objective function has the form $\O((X,d), (c_1, \ldots c_k))=\sum_{x \in X}F(\min_{i \leq k}d(x, c_i))$, for some non-decreasing function $F: \reals \to \reals$.

This family of clustering objectives includes common tasks such as,
\begin{itemize}
\item  $k$-means, where
$\O((X,d), (c_1, \ldots c_k))=\sum_{x \in X}(\min_{i \leq k}d(x, c_i)^2$, 
\item $k$-median,
where
$\O((X,d), (c_1, \ldots c_k))=\sum_{x \in X}\min_{i \leq k}d(x, c_i)$ and 
\item $k$-medoids, where the objective as the same as in $k$-median, but the cluster centers, $c_1, \ldots, c_k$,  are required to be members of $X$.  
\end{itemize}
Given such a clustering, we call the subset of $X$ in the Voronoi cell of each center $c_i$ the $i$'th cluster and denote in by $C_i$.  By extending the format of the objective function to 
$\sum_{i \leq k} G(|C_i|) \sum_{x \in C_i}F(d(x, c_i))$, for some non-decreasing $G$, one can capture some additional common clustering objectives like the sum-of-incluster-distances (MinSum).
In this note we focus on the $k$-means, $k$-median and the $k$-medoids objectives. 

All of these three clustering-motivated discrete optimization problems are known to be NP-hard, and even NP-hard to approximate (some such hardness results are stated quantitatively later). \\

Throughout this paper, we will use $m$ to denote the input size (that is $m=|X|$) and use $D(X)$ to denote the diameter of the input set, namely, $D(X)=\max \{d(x,y): x,y \in X\}$.
When the data is a subset of some Euclidean space, $\reals^n$, we use $n$ to denote that dimension. Furthermore, unless otherwise stated, whenever we discuss the $k$-means objective, we assume that $X \subseteq \reals^n$. 

\subsection{Measures of clustering approximations}
When it comes to approximation algorithms for clustering there is another technical point to be aware of, namely, the way in which one measures the difference between an optimal solution and an approximate one. There are at least two different approaches of quantifying that gap. The first, and probably also the most common one, is to consider only the cost of the solutions. In other words, given a clustering objective function $\O$ an input set $X$ and a clustering, $C$ of it, say that $C$ is an $\epsilon$ cost-approximate good solution if $\O(C,X) \leq Opt(X) (1+\epsilon)$ (alternatively, one could consider additive approximations to the cost, namely, requiring that $\O(C,X) \leq Opt(X) +\epsilon$. Additive approximations arise naturally in the context of statistical machine learning, where approximate solutions are computed based on small samples of the input data). A different type of approximations, more specific to clustering problems, is to define some measure of distance between solutions, such as some distance between the center vectors of two center based clusterings,
say $ \mathcal{D}_{centers}(C, C') \stackrel{def}{=} \inf_{{\pi} \in \Pi}\max_{ i \leq k} d(c_i, c'_{\pi(i)})$, where $\Pi$ is the set permutations of the cluster indices $\{1, \ldots, k\}$,  and $c_1, \ldots c_k$, $c'_1, \ldots c'_k$ are the cluster centers of $C$ and $C'$ (respectively), or $\mathcal{D}_{err}(C, C') \stackrel{def}{=} \left( 1/|X| \right)\inf_{{\pi} \in \Pi} \sum_{i \in \{1, \ldots k\}} |C_i \Delta C'_{\pi(i)}|$.

Having such a measure of distance between clustering solutions, an approximation algorithm is required to come up with a clustering that is close to an optimal clustering w.r.t. that measure. 

There are some implications between these notions of clustering approximations. In particular, note that for, say, the $k$-means objective $|\O(C, X) - \O(C', X)| \leq  \mathcal{D}_{centers}(C, C')$ (for every input set $X$ and clusterings $C, ~C'$). Roughly speaking, approximation is hardest with respect to the $\mathcal{D}_{err}$ distance. Some of the clusterability conditions discussed below imply that such an approximation follows from approximations w.r.t. the other measures. For example, the \emph{Uniqueness of optimum} condition (see below) explicitly requires that clustering solutions that have objective cost close the optimal one are  also close w.r.t. the $\mathcal{D}_{err}$ measure. 
Also, the \emph{Additive perturbation robustness} clusterability implies that any good enough approximation (of an optimal clustering) w.r.t. the $ \mathcal{D}_{centers}$ is in itself an optimal clustering.

\section{Notions of clusterability} \label{notions_clust}
In the past few years there have been several interesting publications along the lines described above,
showing that for various notions of clusterability there are indeed algorithms that find optimal clusterings in polytime for all appropriately clusterable instances. 
Below is a (possibly not exhaustive) list of major notions of clusterability  that have been discussed in that context\footnote{The reader should be aware that different papers use different terminology for similar notions (and similar terminology for different notions), so my choice of terminology below is not always consistent with other publications.}. Most of these definitions can be applied to any of the above mentioned 
clustering objectives. 
\begin{enumerate}
\item \textbf{Perturbation Robustness:} An input data set is perturbation robust if small perturbations of it do not result in a change of the optimal clustering for that set.
\begin{enumerate}
\item Additive perturbation robustness (APR) \cite{AckermanB09}\footnote{The definition of robustness, as well as the implied efficiency of clustering result, in \cite{AckermanB09} are particular cases of a more general definition and more general results of \cite{Ben-David06}. For the sake of conciseness and due to its similarity to other notions discussed below, we present here only this case.}: An input set $(X,d)$ is $\epsilon$-APR if for some optimal $k$-clustering $C$, for every $d'$, if $|d(x,y) -d'(x,y)| \leq \epsilon$  for every $x, y \in X$, then $C \in C_{\O} (X,d')$. Namely, an optimal clustering of the input $(X,d)$ remains optimal for any small (additive) perturbation of this input\footnote{ The definition in \cite{AckermanB09} is formulated as robustness w.r.t. perturbations of the cluster centers of the optimal solution. However, it can be readily seen that the two definitions are equivalent.}. Since this additive condition is not scale invariant, we implicitly add the assumption that the diameter of the input set, $\max_{x,y \in X}d(x,y)$, is at most 1 (otherwise the stability parameter should be multiplied by that diameter).

\item Multiplicative perturbation robustness (MPR) \cite{BiluL10}: An input set $(X,d)$ is $\alpha$-MPR if for some optimal $k$-clustering $C$  such that for every $d'$, if $ ~d(x,y) \leq d'(x,y) \leq \alpha d(x,y)$ for every $x, y \in X$, then $C \in C_{\O} (X,d')$. Namely, an optimal clustering of the input $(X,d)$ that remains optimal for any small (multiplicative) perturbation of this input.\\

\item \cite{BalcanL12} propose the following relaxation of the MPR requirement:
A data set $(X, d)$ is $(\alpha, \epsilon)$-\emph{perturbation resilient } if there exists some optimal $k$-clustering $C$  such that for every $d'$, if $ ~d(x,y) \leq d'(x,y) \leq \alpha d(x,y)$ for every $x, y \in X$, then for some $C' \in C_{\O} (X,d')$, $\mathcal{D}_{err}(C, C') \leq \epsilon$.
\end{enumerate}
\item \textbf{$\epsilon$ -Separatedness:} \cite{OstrovskyRSS12}\footnote{This is a journal version of \cite{OstrovskyRSS06}, where the definition and the main results were initially introduced.} discuss clustering w.r.t. the $k$-means objective. They define an input data set $(X,d)$ to be $\epsilon$-\emph{separated for $k$} if 
the $k$-means cost of the optimal $k$-clustering of $(X,d)$ is less then $\epsilon^2$ times the cost of the optimal $(k-1)$-clustering of $(X,d)$.

\item \textbf{Uniqueness of optimum}: \cite{BalcanBG13}\footnote{This is a journal version of \cite{BalcanBG09}, where the definition and the main results were initially introduced.} define a data set to be $(c, \epsilon)$\emph{-approximation-stable} with respect to some \emph{target clustering} $C_T$ if every clustering $C$ of $X$
whose objective cost over $(X,d)$ is within a factor $c$ of the objective cost of $C_T$ (on $(X,d)$) is $\epsilon$-close to $C_T$ w.r.t. some natural notion of between - clustering distance. It is easily seen that such a condition holds with respect to any $C_T$ if and only if it holds (up to constant factors) w.r.t. the optimal clustering for $(X,d)$ (see Fact 2.2. of \cite{BalcanBG13}). We relate to this property as Uniqueness of Optimum since it rules out the possibility of having two significantly different close-to-optimal-cost solutions.

\item \textbf{$\alpha$-center stability}: \cite{AwasthiBS12} define an instance $(X,d)$ to be \emph{$\alpha$-center stable } (with respect to some center based clustering objective $\O$) if for any optimal clustering $C \in C_{\O} (X,d)$ defined by centers $c_1, \ldots c_k$ (of the clusters $C_1, \ldots C_k$ respectively), for every $ i \leq k$ and every $x \in C_i$, and every $j \neq i$,
$\alpha d(x, c_i) < d(x, c_j)$. Namely, points are closer by a factor $\alpha$ to their own cluster center than to any other cluster center.

\item \textbf{$(1+\alpha)$ Weak Deletion Stability:} \cite{AwasthiBS10} define an instance for $k$-clustering to satisfy the \emph{$(1+\alpha)$ Weak Deletion Stability} condition if, for all $i \neq j$,
$$OPT^{(i \to j)} > (1+ \alpha) OPT,$$ where $OPT$ is the cost of the optimal clustering of that instance, and, if the optimal clustering is determined by centers $(c_1, \ldots, c_k)$ and the optimal clusters are $(C_1, \ldots, C_k)$, then $OPT^{(i \to j)}$ is the cost of the clustering obtained by by removing the center $c_i$ and assigning all the points in $C_i$ to the center $c_j$. Note that an instance for $k$-clustering is $\epsilon$-separated, then it  satisfies the  $\epsilon^2$-WDS condition for that $k$.

\end{enumerate}

As varied as the above list of proposed notions may sound, it turns out that almost all (except for the additive perturbation robustness, which is also the only one that does not yield efficiency for large $k$) imply that data satisfying them is structured such that the vast majority of the data points can be assigned to compact clusters that are very widely separated (or that all but a small fraction of the clusters are such). We provide quantitative versions of this claim in Section \ref{lower_bounds}. In fact, this common characteristic  of the notions is the main feature that is being used in showing that, under such conditions, clustering can be carried out efficiently.

\section{To what extent do the notions meet the requirements listed above?} \label{meet_req}

While all of the above notions sound intuitively plausible (concrete arguments supporting that plausibility can be found in the papers presenting them), the quantitative values
of the clusterability assumptions are essential for evaluating that plausibility. We shall see below that the currently known results concerning these notions yield the desired efficiency of computation only when the clusterability parameters are set to values that are beyond what one might expect practical inputs to satisfy.

To keep this note focused, we provide a relatively high level view of some of the major relevant results. However, since 
the actual values of the 
parameters (that define the clusterability notions) determine both the runtime of the algorithms and the restrictiveness of the clsuterability conditions, these concrete values are needed when we wish to evaluate and the gap between what we currently know and the optimistic CDNM thesis.

\subsection{Computational efficiency of clustering clusterable inputs} \label{pos_comput}

An important distinction in this context concerns the meaning of hardness of computation. Clustering tasks where the clusters are determined by selecting cluster centers from the input set can clearly always be solved in time $m^k$ (where $m$ is the input size and $k$ is the number of clusters), by performing an exhaustive search over all possible cluster centers. For such problems, the term "feasible" usually refers to run time bounded by a polynomial in both $m$ and $k$. On the other hand, tasks like $k$-means, where the input set resides in some euclidean space, $\reals^n$, and cluster centers can be arbitrary points in that space, are often NP hard already for fixed values of $k$ (e.g., $k=2$) when one takes the space dimension $n$ as a parameter that the runtime is a function of. For such problems, algorithms that have polynomial dependence on $m$ and $n$ may be considered ``feasible" even if they have exponential dependence on $k$. Of course, in order to have solutions that are also polynomial in $k$, the requirements on the input instances are more demanding. \\

We summarize the main relevant results according to the different notions of clusterability that they require from the input instances (we let $m$ denote the size of the input set $X$ and $k$ the target number of clusters);
\begin{enumerate}

\item  \textbf{Additive perturbation robustness (APR)}: \cite{AckermanB09} show that for every center-based clustering objective and every $\mu >0$ there exists an algorithm that runs in time $O \left( m^{k/\mu^2} \right)$ and finds the optimal $k$-means clustering for every instance that is $\mu$-APR. Using the results of \cite{Ben-David07} the parameter $m$ in the runtime can be replaced by $\frac{nk}{\mu^2 \epsilon^2}$ if 
one settles for a solution whose cost is at most $OPT(X) +\epsilon |X| D(X)$ (recall that $D(X)$ is the diameter of the input set). Recalling that for fixed $k$ the $k$-means problem is NP hard when the Euclidean dimension, $n$, is considered an input parameter, the $\mu$-APR condition allows to get rid of the dependence of $n$ in the runtime, and replace it by dependence on the robustness parameter $\mu$.

If we allow $\mu$ to depend on $m$, we get runtime poly $(m)$, as long as $k/\mu^2 = O(\log m/\log \log m)$ and $1/\epsilon$ and $n$ are upper bounded by poly$\log m$ and poly $m$, respectively. 

\item  \textbf{Multiplicative perturbation robustness (MPR)}: \cite{AwasthiBS12} show that for every $\epsilon>0$ there exists an algorithm that finds an optimal solution to the $k$-median clustering problem for all inputs that are $(3-\epsilon)$- MPR in time $O \left( m^2+mk^2\right)$. \cite{BalcanL12} improve these results to assuming only $(1+\sqrt{2})$-MPR as well as obtaining similar results for the Min-Sum objective. They also prove an efficient approximation result under the weaker assumption that the optimal input clustering is an \emph{approximation} for the optimal of any multiplicative $\alpha$ perturbation of that input. \cite{BiluDLS13} show that for the Max-Cut objective (which considers clustering into $k=2$ clusters), there exist
algorithms that find the optimal solution for any $\sqrt{m}$-MPR input in time polynomial in $m$ (they also show the existence of efficient algorithms for solving Max-Cut under other data assumptions. However, as the focus of this note are center-based clustering tasks, we do not elaborate on those results).\\

Furthermore, \cite{BalcanL12} show that there is a polynomial time algorithm that for any $\alpha > 2+\sqrt{7}$, and any $(\alpha, \epsilon)$-\emph{perturbation resilient } (for the $k$-median objective) input for which the smallest cluster in the optimal clustering contains at least $5\epsilon m$ many points, finds a clustering that is $\epsilon$-close to an optimal clustering.

\item \textbf{$\epsilon$ -Separatedness:} \cite{OstrovskyRSS12} focus of the $k$-means objective. They show propose a variant of the Lloyd algorithm that,  for $k=2$ assuming $\epsilon$-separatedness of the input,  run in time linear in $m$ and $n$ (the Euclidean dimension) and yields a clustering solution $C$ that  with probability $(1- O(\rho))$, has cost  $\O(C,X) \leq \frac{OPT(X)}{1-\rho}$ where $\rho=\Theta(\epsilon^2)$. For the $k$-means problem for arbitrary $k$, they get, under the same assumption, a (different) variant to the Lloyd algorithm that yields  a clustering solution $C$ that, with probability $(1- O(\sqrt{\epsilon}))$, has cost  $\O(C,X) \leq OPT(X) \frac{1-\epsilon^2}{1-37\epsilon^2}$  in time $O(mkn+k^3n)$.\\

\item \textbf{Uniqueness of optimum}: \cite{BalcanBG13}  propose algorithms that, for data sets that are $(1+\alpha, \epsilon)$\emph{-approximation-stable} find, in time polynomial in $m$ and $k$ clusterings that are $O(\epsilon + \epsilon/\alpha)$ close (w.r.t. the between-clusterings distance $\mathcal{D}_{err}$) to the optimal clusterings w.r.t. the $k$-means and w.r.t. the $k$-median objectives.

 \item \textbf{$\alpha$-center stability}: \cite{BalcanL12} present an algorithm that, for any $\alpha \geq 1+\sqrt{2}$ outputs an optimal $k$-median clustering, as well as a binary hierarchical clustering tree for which the optimal $k$-means clustering is a pruning of that tree, in time polynomial in $m$ and $k$.\\

\item \textbf{$(1+\alpha)$ Weak Deletion Stability:} For the $k$-means objective, \cite{AwasthiBS10} propose an algorithm that given any positive $k$, $\epsilon$ and $\alpha$, for any input $X$ satisfying the $(1+\alpha)$ Weak Deletion Stability condition it finds a clustering $C$ such that $\O(   C   ) \leq (1+\epsilon) OPT(X)$ in time $m^{O(1)}(k \log m )^{poly(1/\epsilon, 1/\alpha)}$.

\end{enumerate}

\subsection{How restrictive are the clusterability parameters required for the efficiency of computation results?} \label{lower_bounds}

In this subsection we examine the clusterability conditions listed in terms of the degree of separation between clusters that these conditions require (in the optimal clustering) when their clusterability (or data niceness) parameters assume values that  suffice for showing efficiency of the corresponding proposed clustering algorithms. To measure those cluster-separation requirements, we focus of the relationship between the average distance of a point to the center of its cluster, $AvDis=OPT/m$ (which can be thought of as the average ``cluster radius" in an optimal clustering), to the minimal distance of a point from any other center, $w_2(x) = \min_{i \neq i(x)}d(x, c_i)$, where $c_1, \ldots, c_k$ are the cluster centers in an optimal clustering of the given data set and  $i(x)$ is the index of the cluster a point $x$ belongs to.  As we shall argue below, most of the results cited above require a rather large value of $w_2(x)/AvDis$, for most of the points $x$ in the input set.

\begin{enumerate}
\item \textbf{Perturbation Robustness}:
\begin{enumerate}
\item \textbf{Additive perturbation robustness (APR)}: 
The first point to note about the efficiency results of \cite{AckermanB09} is that they focus on the case of fixed number of clusters and therefore their runtime upper bounds are not polynomial in $k$. Furthermore, although that run time is polynomial for any fixed $k$, the degree of that polynomial is impractically high, $k/\mu^2$, where $\mu$ is the robustness parameter. It is also worthwhile noting that  these efficiency results are shown only for data residing in any Hilbert space, it is not known if they extend to data in arbitrary metric spaces.

\item \textbf{Multiplicative perturbation robustness (MPR)}:  It is not difficult to see that for any $\alpha$ a data set that is $\alpha$-MPR is also $\alpha$-center stable (see, e.g.,  \cite{BalcanL12}). In fact, most of the efficiency of clustering results for data satisfying MPR conditions actually use only the implied center stability properties.
We will see below hardness results for clustering under center stability conditions. While not implying hardness for clustering under MPR, they do show inherent limitations of the proof techniques (and algorithms) used so far for clustering under this clusterability assumption.

How restrictive is the requirement of $2$-center stability for real data? For concreteness, consider the very simplistic assumption that the data is nicely confined to $k$ balls, $B(c_1, r_1) , \ldots B(c_k, r_k)$, where $(c_1, \ldots, c_k)$ are the centers of the clusters in the optimal $k$-clustering and the $r_i$'s are the radii of these balls. Such data satisfied the $2$-center stability requirement if (and only if) for every $i \neq j$, $d(c_i, c_j) \geq 3\max \{r_i, r_j\}$.
When the clusters are not ball shaped, the requirement may become more complicated. In particular, in Euclidean spaces (considering again the optimal $k$-clustering),  denoting by $r_{i,j}$ the distance from $c_i$ of the furthest point in the cluster $C_i$ along the line segment connecting $c_i$ an d $c_j$, the $2$-center stability requirement  implies that $d(c_i, c_j) \geq 3r_{i,j}$ for every $i\neq j$.

The $(\alpha, \epsilon)$-perturbation resiliency condition relaxes this requirement by allowing for some points to fail the strict requirement ``\emph{$\alpha$ times closer to your own center than to any other center}". However, the efficiency of clustering results under this condition (\cite{BalcanL12}) apply only when the number of such violations does not exceed the number of points in the smallest cluster. In particular, for every $k$ the fraction of violations parameter $\epsilon$ is upper bounded by $1/k$, shrinking to zero as $k$ grows.
\end{enumerate}

\item \textbf{$\epsilon$ -Separatedness:}

In Section  \ref{pos_comput}, the results are cited the way they appear in  \cite{OstrovskyRSS12}. To evaluate how strict are the separateness conditions required for the efficiency results there, we take closer look at the actual constants behind the asymptotic notation; The parameter $\rho$ equals $\frac{100\epsilon^2}{1-\epsilon^2}$.
This implies that in order to have any significant success probability in the above results, $\epsilon^2$ should be at most $1/200$. In other words, the benefits of the $\epsilon$-separatedness condition kick in only when the cost of optimal $k$-means clustering is at most $1/200$ times the cost of the optimal $(k-1)$-clustering. Furthermore, the big $O$ notation in these results hide constant factors that make those parameter settings even more demanding. 

Lemma 3.1 of that paper may help to better appreciate how severe are such requirements. That lemma states that for the 2-means problem, for $\epsilon$-separated inputs, in the optimal clustering, the average distance of data points to their centers  is less than $O(\epsilon^2)$ times the distance between those centers. This lemma can be readily extended to $k$-mean clustering for any $k >1$.Namely,
\begin{lemma}
Let $k$ be at least 2 and let $X$ be any subset of euclidean space that satisfies the $\epsilon$ -Separatedness condition for $k$-means. Let $C$ be an optimal $k$-means clustering of $X$ and $c_1, \ldots, c_k$ its cluster centers. Finally, let $r_i$ denote the mean square distance of the points in the $i$'th cluster from their center, $c_i$. Then for any $i  \leq k$,$$r_i \leq \frac{\epsilon^2}{1- \epsilon^2}\min_{j \neq i}||c_i -c_j||.$$ 

\end{lemma}

In other words, in order satisfy the $\epsilon$-separatedness clusterability condition, with a parameter $\epsilon$ that suffices to guarantee success of the \cite{OstrovskyRSS12} proposed algorithm,
the data must be organized in small clusters that are extremely well separated. Under such conditions, it is not surprising that a sampling distribution that aims to pick a set of pairwise far points end up picking a representative residing in different clusters.

\item \textbf{Uniqueness of optimum}: 
To appreciate the tradeoffs between data niceness requirements and the efficiency of the clustering algorithm (of \cite{BalcanBG13} ) on such data, it is worthwhile to review Lemma 3.1 of that paper. The lemma examines the implications of the $(c, \epsilon)$-approximation-stability assumption on the degree of separations between clusters in data satisfying that assumption. 

\begin{lemma}[\cite{BalcanBG13}]\footnote{While \cite{BalcanBG13} phrases its results w.r.t. to some "target clustering" $C_T$, here, for the sake of concreteness, and easier comparison with the other papers discussed, we consider the case where that target clustering is an optimal clustering w.r.t. the relevant clustering objective. The results w.r.t. a different target clustering are essentially the same up to an additive term of $\epsilon^{\*} = \mathcal{D}_{err}(C_T, C_{OPT})$.} \label{appr_stab_sep} If an instance $X$ satisfies the $(1+\alpha, \epsilon)$-approximation-stability condition for the $k$-median objective, then
\begin{enumerate}
\item In the optimal clustering of $X$,
all but $6 \epsilon m$ of the data points satisfy $w_2(x) -w(x) \geq \frac{\alpha AvDis}{2\epsilon}$.
\item For any $t>0$, at most $t\epsilon m/\alpha$ many points have $w(x) \geq \frac{\alpha AvDis}{t \epsilon}$.
\end{enumerate}
\end{lemma}

The main issue with the efficiency results under this condition is the constants implicit in the $O(\epsilon)$ notation of those results. The proof of Theorem 3.9 there shows that the (under $(1+\alpha, \epsilon)$-approximation-stability condition) the algorithm gets a $4b$-approximation of the optimal (or target) clustering, where $b \leq (6 + 40/\alpha)\epsilon m$.
Since any clustering is trivially an $m$-approximation, the result is only meaningful once $(6 + 40/\alpha)\epsilon <  < 1$, and in particular,  $\alpha/\epsilon > 40$.
However, in light of Lemma \ref{appr_stab_sep}, this implies that for the vast majority of the points $x$ in the input set, $w_2(x) -w(x) \geq 20 AvDis$ - a rather strong between-clusters-separation requirement. 
\item \textbf{$\alpha$-center stability}: 
This is probably the condition for which our theoretical understanding is most complete. On one hand we have the \cite{BalcanL12} efficiency  result for $\alpha > 1 +\sqrt{2}$, and on the other hand there is an almost matching lower bound:

\begin{theorem}[\cite{ShalevReyzin14}]\footnote{This is a journal version of  \cite{Reyzin12} where the result initially appeared.}
For any $\epsilon >0$ the problem of finding the optimal $k$-median clustering for $(2-\epsilon)$-center stable inputs is NP-hard.
\end{theorem}
It is worthwhile to note that 
this results addresses the setup in which $k$ is part of the input. It does not imply NP-hardness for the problem for any fixed number of clusters. Furthermore, it is obtained using a metric that is not Euclidean. For data in Euclidean spaces a similar result probably applies with a somewhat lower value of $\alpha$.

Another relevant result of  \cite{ShalevReyzin14}  is that once the parameter $\alpha$ exceeds $2+\sqrt{3}$, data satisfying the $\alpha$-center stability condition is somewhat trivial.
For such data sets,  for any $x,y, z$, whenever $x,y$ are in the same cluster and $z$ is in a different cluster (w.r.t. an optimal clustering) then $d(x,y) < d(x,z)$ (this is called \emph{perturbation resiliency)}. This property allows a simple dynamic programming algorithm to find the optimal clustering in time $O(m^2)$.

In conclusion, from the viewpoint of  $\alpha$-center stability, there is relatively little gap between being NP-hard and being (almost) trivially clusterable.

\item \textbf{$(1+\alpha)$ Weak Deletion Stability:} The \cite{AwasthiBS10} bound on the running time of the algorithm has only polynomial explicit dependence on the number of clusters $k$. However, it has exponential dependence on the niceness parameter $1/\alpha$ (the deletion stability requirement becomes less restrictive with smaller $\alpha$).
The following claims address the relationship between that parameter and the number of clusters. Our conclusion is that, as long as the clusterability requirement, parameterized by $\alpha$, is not extremely strong, the running time formula of \cite{AwasthiBS10} is, in fact, exponential in $k$.

\begin{claim} \label{claim1_WDS} If $(X,d)$ is $(1+\alpha)$ WDS, then $1/\alpha > k\frac{OPT}{md_{min}}$. 
\end{claim}

\noindent \textbf{Runtime implications:} Recall that $\frac{OPT}{m}$ is the average of the square distance between data points and the centers of their clusters in the optimal clustering. Since  
 $1/\alpha$ is in the exponent of the runtime bound, it follows that as long as the ratio between the average distance of points to their cluster centers and the minimum distance between the centers (of the optimal clustering) does not grow  superpolynomially with the number of clusters, $k$, the runtime bound is, in fact, exponential in $k$.

\begin{proof}[Proof of Claim \ref{claim1_WDS}] Let $C=(C_1, \ldots, C_k)$ be an optimal clustering of $(X,d)$. Note that,  for every $i \leq k$, if $c_j$ is the closest center to $c_i$ then $OPT^{(i \to j)} \leq OPT+ |C_i|d_i$ (since by assigning the points of $C_i$ to some center $c_j$
the cost associated with each point of $C_i$ grows by at most $d(c_i, c_j)=d_i$). Pick $i$ such that $|C_i| \leq 1/k$ and $d_i=d_{min}$ (such $i$ exists since at least one of the clusters contains at most $m/k$ points). It follows that for such an $i$, for $j \in \argmin\{d(c_i, c_j)\}$, $OPT^{(i \to j)} \leq OPT+ md_{min}/k$. The $(1+\alpha)$ WDS property of $(X,d)$ therefore implies that $md_{min}/k > \alpha OPT$, which is equivalent to the inequality that the claim states.

\end{proof}
Furthermore, \cite{AwasthiBS10} show the following similar manifestation of the strong implications on the $(1+\alpha)$ WDS condition, in terms of the lower bounds it implies on between-cluster-centers distances:

\begin{claim} 
For any $(1+\alpha)$ WDS $k$-median instance, for any center $c_i$ of its optimal clustering and any data point $x \notin C_i$,
\[d(x,c_i ) \geq \frac{\alpha OPT}{2|C_i|}\]
and 
for any $(1+\alpha)$ WDS $k$-means instance, for any center $c_i$ of its optimal clustering and any data point $x \notin C_i$,
\[d^2(x,c_i) \geq \frac{\alpha OPT}{4|C_i|}\]
\end{claim}

\noindent \textbf{Discussion}: Rewriting these bounds (for concreteness, the bound for $k$-median) in terms of the the average distance of a point in $X$ from the center of its cluster, $AvDis=\frac{OPT}{m}$,
it reads $d(x,c_i ) \geq\alpha \frac{m}{2|C_i|} AvDis$. Since  for every $k$ and every $t$ there are at most $k/t$ clusters of size $> mt/k$. In particular, there are at most $\log (k)$ many clusters of size $> m /\log(k)$. It follows that for all but $\log(k)$ of the $k$ clusters $C_i$, for every point $x$ outside the cluster, $d(x, c_i) \geq 0.5\alpha  \log(k) AvDis$.
In other words, if one considers the case of large $k$ (which is the source of computational difficulty that the clusterability condition is aimed to overcome), for any fixed $\alpha$, for any instance satisfying the $(1+\alpha)$ WDS condition, the vast majority of the clusters are so distant from the rest of the data points that any point outside such a cluster is further from the center of that cluster by a factor of $\log(k)$ compared to the "average radius" of the clusters.

\end{enumerate}

\subsection{Efficient testability of the clusterability conditions}
When it comes to testing whether a given clustering instance satisfies any of the above clusterability conditions, a key point to note is that they are all phrased in terms of condition pertaining to the optimal clustering of the given data. Finding such optimal clusterings is NP-hard. Furthermore, as far as I am aware, there exist no efficient algorithm for testing, given a data set $(X,d)$ and a $k$ clustering of it, $C$, whether $C$ is an optimal clustering for $(X,d)$ (say, w.r.t. either the $k$-means or the $k$-median objective). I therefore conjecture
that testing each of the conditions we have discussed here is NP-hard. 

Some of those conditions can be also phrased as a niceness property of a given clustering (rather than a property of the data). For example, 

\begin{quote} Given  a $k$ clustering $C$ for an instance $(X,d)$, defined by a vector of centers, $c_1, \ldots c_k$,   say that  \emph{$C$ is $\alpha$-center stable}  if for every $ i \leq k$ and every $x \in C_i$ and $j \neq i$,
$\alpha d(x, c_i) <d(x,c_j)$. 

\end{quote}

However, it is easy to see that a clustering that satisfies such a property in not necessarily optimal, and that the fact that $(X,d)$ allows such a clustering does not imply that it is nicely clusterable; As a simple example, consider 2-means for an instance $X \subseteq \reals^2$ that consists of 1000 points, 999 of them  evenly spread in the unit ball and the last point 
at $(0, 50)$. The clustering that has the unit ball as one cluster, and the outlier point as the other (singleton) cluster, satisfies the given-cluster version of $\alpha$-center stability for $\alpha= 50$. However, this data set is not $\alpha$-center stable for any $\alpha >1$.

In fact, the positive results of \cite{BalcanL12}, showing efficient clustering algorithm for $(1+ \sqrt{2})$-center stable data, can be rephrased as follows:
There is an efficient algorithm that when applied to any $(1+ \sqrt{2})$-center stable data set, outputs a binary hierarchical cluster tree such that every  $(1+ \sqrt{2})$-center-stable clustering $C$ of that data set is the result of some pruning of that tree. Given such a tree, for any feasibly computable objective function, the tree can be efficiently searched to find its minimum cost pruning w.r.t. this objective.

The niceness condition concerning the input data is only invoked to show that the minimum cost pruning of the tree is also a minimum cost clustering of that data set.

Similarly, one can define, for a given $k$ clustering $C$ of a set $(X,d)$,  when \emph{$C$ is $(1+ \alpha)$-weakly-deletion stable}. Once again, for every value of $\alpha$, there are examples of data sets for which there are clustering? that are $(1+ \alpha)$-weakly-deletion stable, and yet are not optimal $k$-means (or $k$-median) clusterings.
However, it is not clear to me if for arbitrarily large values of $\alpha$ there exist instances $(X,d)$ that have a clustering $C$ that is $(1+ \alpha)$-weakly-deletion stable and yet 
$(X,d)$ is not $(1+ \alpha)$-weakly-deletion stable.

An easier goal than coming up with a useful notion of clusterability that is efficiently testable, is to come up with a notion of niceness of a given clustering $C$, such that one can efficiently test if a clustering solution $C$ for an instance $(X.d)$ satisfies that requirement, and so that if it does, it is guaranteed to be an optimal clustering for the $(X,d)$. As far as I am aware no such notion currently exits. As noted above, it seems to be an open question whether the notion of a clustering $C$ being $(1+ \alpha)$-weakly-deletion stable
(for some sufficiently large $\alpha$) implies that the domain set of such a clustering is necessarily $(1+ \alpha)$-weakly-deletion stable.

\subsection{Implications for common practical clustering algorithms}
Among all the works surveyed in this note, only one, the results of \cite{OstrovskyRSS12}, address (a feasible variant of) a practical algorithm - the popular Lloyd clustering algorithm.  It would be very interesting to come up with results showing that  some popular clustering algorithm (or an application of a practical approximation algorithm) efficiently yield  guaranteed good quality clusterings, under some other, or more relaxed, niceness of data conditions. The recent work of  \cite{AwasthiBC14}
can be viewed as a step in that direction. They ask under which separation condition do various convex relaxations exactly recover the ``correct" clustering.
However, that work addresses a different version of clustering problems, in which one assumes that the data is generated by some parameterized generative model (a balanced mixture of spherical Gaussians, in the case of that paper), and aims to recover those parameters.

\section{Conclusions} \label{conclusions} Several notions of clusterability have been proposed so far. Depending on the values of the parameters defining those notions, each of them ranges from being very lenient to a  highly constraining data requirement.
For each notion there is a parameter range so that, for data conforming to the clusterability requirement in that range, an optimal clustering can be rather trivially found. 
Clusterability with parameter values that suffice for the currently available efficient clustering results turns out to be rather strong requirements, that eminently restricts the practical significance of the currently available results.

The current failure to support the CDNM thesis can stem from various sources. First, of course, maybe the thesis is just false. My personal belief is that, while it may very well be the case that some practical clustering tasks are indeed computationally hard for some real data instances, there are many more cases where data of practical interest does yield not-too-hard-to-find meaningful clusterings (though, of course, most of the time we have no way of knowing whether those are optimal clusterings in any formal sense of optimality).

Another explanation to the shortcomings of current results is that they may just be an artifact of the algorithms and proof techniques that we currently have. Maybe one could eventually come up with efficient algorithms that will cluster well under much less restrictive parameter settings of the clusterability notions listed in this note. Indeed, for most of those notions we do not have any close-to-matching computational hardness lower bounds. I doubt if that is indeed the case. As mentioned above, for the notion of $\alpha$-center stability, the gap between the parameter values sufficient for efficient clustering and those that imply NP hardness is very small, $(1+\sqrt{2})$ vs $2$. Furthermore, the $\alpha$-center stability is a central notions, in the sense that almost any other of the notions of clusterability discussed above implies it (or some variants of that condition), and the current results rely on those implications for proving the efficiency of clustering under those conditions.

I believe that part of the answer is that we have not yet discovered the appropriate notions of clusterability. In light of the results surveyed in this paper, I think that notions of clusterability that aim to substantiate the CDNM statement should not be just a way of formalizing large between-clusters separation. Apparently,  as demonstrated above, such assumptions become too restrictive before they yield efficient clustering results.

Finally, in the last paragraph of concluding remarks below, I would like to argue that if we really wish to model clustering as it is required and used in applications, the formulation of clustering tasks as computational problems should be revisited and revised.
\subsection{Call for a change of perspective on the complexity of clustering}
All the papers surveyed above, as well as most of the current theoretical work on the computational complexity of clustering, focus on \emph{concrete clustering objectives} aiming to find the best clustering for a \emph{given number of clusters}.
However, the practice of clustering is widely varied. There are applications, like clustering for detecting record
duplications in data bases (say, records of patients from various hospitals and clinics), where the user does not set the number of clusters in advance, and aims to detect sets of mutually similar items to the extent that such sets occur in the input data. 
In other applications, like vector quantization for signal transmission or facility location tasks, while the objective function is usually fixed (say, $k$-means), 
 there is no implicit ``target clustering" and the usefulness of a resulting clustering is not diminished by having various different close-to-optimal solutions. 
 In some such applications $k$ is externally determined,  however, it is also common to consider optimizing some ``compression vs distortion" tradeoffs, rather than aiming for a fixed number of clusters. 

Furthermore, while the restriction of the problem of finding a good clustering to a given number of clusters $k$ may make practical sense when $k$ is small, for data sets that yield a very large number of clusters it is harder to imagine realistic situations in which that number, $k$, should be fixed independently of the particular input data set.
Still, most of the work surveyed above focuses on analyzing the asymptotic, w.r.t. $k$, computational complexity of $k$-clustering where $k$ is determined as part of the problem input. 

In many cases, the actual goal of clustering procedures is to find \emph{some} meaningful structure of the given data, and is not committed to any fixed objective function or any fixed number of clusters. The currently available theoretical research does not provide satisfactory formalizations of such ``flexible" clustering tasks\footnote{There haas been some recent work theoretically analyzing a notion of statistical stability (with respect to independent samplings) as a tool for determining an appropriate number of clusters as a function of the input data, e.g., \cite{BDPS07}, \cite{ShamirT10}. However, the conclusions if this work are mainly negative, showing that some proposed approaches may not work as intended.}, let alone an analysis of their computational complexity.
It may well be the case that our intuition of \emph{clustering being feasible on practically relevant cases} stems from clustering tasks that do not fit into the rigid fixed-$k$-fixed-objective framework of clustering.\\

\section{Some followup open problems}
Of course, in view of the results we have surveyed, the most obvious open problem is still the status of the CDNM thesis. The  challenges referred to in the above Conclusions section 
may also be viewed as ``open problems". However, in this section, we wish to list some more technical and concrete problems whose answers will advance our understanding of the main topics of this paper.

\begin{enumerate}
\item \textbf{Lower bounds under the above clusterability assumptions:} So far, it seems that the only notion of clusterability for which we have at this point meaningful lower bounds
(on the computational complexity of finding an optimal clustering for data satisfying the condition) is the $\alpha$-center stability. Even for that notion, the lower bounds of \cite{ShalevReyzin14} require the input data to be an arbitrary metric space and do not apply to data in a Euclidean space. An obvious, though of relatively minor significance question is: For which values of $\alpha$ does the problem of optimizing the $k$-means or $k$-median objectives for $\alpha$-center stable instances become NP hard  for instances in Euclidean space, or for instances in $\reals^2$?
More significant open questions are finding parameter values for which  the other notions of clusterability become NP hard, and zooming in on the range of parameter values for which clustering under those clusterability conditions can be feasibly carried out. 

\newpage
\item \textbf{Clustering via linkage based hierarchical clustering trees:} The algorithm of \cite{BalcanL12} is based on a \emph{linkage-based}\footnote{Linkage-based clustering algorithms are algorithms that, given some clustering instance $(X,d)$ define a notion of dissimilarity over subsets of $X$, $\hat{d}$, and then contract a tree whose nodes are labeled by subsets of $X$ as follows: The leaves are all the singleton sunsets $\{x\}_{x \in X}$, and repetitively, it picks a pair of node subsets $A_i,B_i$ minimizing the dissimilarity $\hat{d}(A,B)$ over all node subsets that have already been generated  and creates a parent node, labeled $A \cup B$ above these two nodes, until a (root node) labels $X$ is reached. See \cite{AckermanBL10} for a more detailed discussion.} agglomerative construction of a cluster tree that is guaranteed to have any $(1+\sqrt{2})$-center stable clustering of the input data as a pruning of that tree. For which other notions of clustering niceness can one have an appropriate, feasibly computable, linkage based clustering that is guaranteed to output
a tree with such a properly (that is, every ``nice" clustering is a obtained by some pruning of the tree)?

\item \textbf{The relationship between being a nice clustering and optimizing common clustering objectives:} Which notions of clustering niceness imply that 
\begin{enumerate} 
\item If a data set $(X,d)$ allows such a ``nice clustering", then that clustering is bound to be an optimal $k$-means (or $k$-median) clustering of  $(X,d)$.
\item If a data set $(X,d)$ allows such a ``nice clustering" then there must be an optimal $k$-means (or $k$-median) clustering of  $(X,d)$ that is a nice clustering (for that notion of niceness).
\item If a data set $(X,d)$ allows such a ``nice clustering" then it is unique (namely, there exist no other similarly nice clustering of $X,d)$).
\end{enumerate}

\item \textbf{Applying common approximation techniques to clustering optimization problems:} Pick any common approximation technique (like linear programming relaxations) 
and come up with some naturally sounding notions of clusterability under which such an approximation algorithms is guaranteed to find the optimal clustering (rather than just approximating it) efficiently. Under which clusterability conditions will the approximations guaranteed for such algorithms be better than known hardness approximation lower bounds for clustering arbitrary instances?

\end{enumerate}

\section*{Acknowledgements} I am grateful to Shalev Ben-David,  Lev Rayzin  and Ruth Urner  for insightful discussions concerning this paper. \\

\bibliographystyle{plain}
\bibliography{Researchbib.bib}

\end{document}